\documentclass[11pt]{article}

\usepackage{fullpage}
\usepackage{amsmath}
\usepackage{amsthm}
\usepackage{amssymb}
\usepackage{amsfonts}
\usepackage{hyperref}
\usepackage{color}
\usepackage{xcolor}
\usepackage{arydshln}
\usepackage{booktabs}

\usepackage{graphicx} 
\graphicspath{ {./images/} }
\usepackage[rightcaption]{sidecap}
\usepackage{wrapfig}

\usepackage{enumitem}
\setlist{itemsep=-3pt}

\newcommand{\ignore}[1]{}

\newtheorem{theorem}{Theorem}

\newtheorem{lemma}{Lemma}

\newtheorem{claim}{Claim}

\newcommand{\A}{{\cal A}}
\newcommand{\B}{{\cal B}}

\newcommand{\brank}{{\rm br}}
\newcommand{\br}{{\rm br}}

\newcounter{counter}
\newcommand{\correct}[1]{\addtocounter{counter}{1}\color{green}$$\mbox{{\Huge set \arabic{counter}}}$$
\color{red}$$\mbox{{\huge (1) Correctness (2) Consistency}}$$
$$\mbox{{\Large (3) Google Scholor (4) Yellow (5) Grammarly pdf}}$$\color{blue} }

\begin{document}

\title{A Note on Property Testing of the Binary Rank}
\author{{\bf Nader H. Bshouty}\\ Dept. of Computer Science\\ Technion, Haifa, Israel.\\
}

\maketitle
\begin{abstract}
Let $M$ be a $n\times m$ $(0,1)$-matrix. We define the $s$-binary rank, $\br_s(M)$, of $M$ to be the minimal integer $d$ such that there are $d$ monochromatic rectangles that cover all the $1$-entries in the matrix, and each $1$-entry is covered by at most $s$ rectangles.
When $s=1$, this is the binary rank,~$\br(M)$, known from the literature. 

Let $R(M)$ and $C(M)$ be the set of rows and columns of~$M$, respectively. We use the result of Sgall~\cite{Sgall99} to prove that if $M$ has $s$-binary rank at most~$d$, then $|R(M)|\cdot |C(M)|\le {d\choose \le s}2^{d}$ where ${d\choose \le s}=\sum_{i=0}^s{d\choose i}$. This bound is tight; that is, there exists a matrix $M'$ of $s$-binary rank $d$ such that $|R(M')|\cdot |C(M')|= {d\choose \le s}2^{d}$.

Using this result, we give a new one-sided adaptive and non-adaptive testers for $(0,1)$-matrices of $s$-binary rank at most $d$ (and exactly $d$) that makes $\tilde O\left({d\choose \le s}2^d/\epsilon\right)$ and $\tilde O\left({d\choose \le s}2^d/\epsilon^2\right)$ queries, respectively.

For a fixed $s$, this improves the query complexity of the tester of Parnas et al. in ~\cite{ParnasRS21} by a factor of $\tilde \Theta (2^d)$. 
\end{abstract}

\section{Introduction}
Let $M$ be a $n\times m$ $(0,1)$-matrix. We define the $s$-binary rank, $\br_s(M)$, of $M$ to be the minimal integer~$d$ such that there are $d$ sets (rectangles) $I_k\times J_k$ where $I_k\subseteq [n]:=\{1,\ldots,n\},J_k\subset [m],k\in [d]$ such that\footnote{For $M$, the $(i,j)$ entry of the matrix is denoted by $M[i,j]$.} $M[i,j]=1$ for all $(i,j)\in I_k\times J_k$, $k\in [d]$ (monochromatic rectangles), and for every $(i,j)\in [n]\times [m]$ where $M[i,j]=1$, there are at least one and at most $s$ integers $t\in [d]$ such that $(i,j)\in I_t\times J_t$ (each~$1$-entry in $M$ is covered by at least one and at most $s$ monochromatic rectangles).
When $s=1$, $\br_1(M)$, is the binary rank,~$\br(M)$, and when $s=\infty$, $\br_\infty(M)$ is the Boolean rank. Both are known from the literature. See, for example,~\cite{GregoryPJL91}. 

The binary rank can also be defined as follows. The binary rank of a $n\times m$ $(0,1)$-matrix $M$ is  equal to the minimal $d$, where there are $n\times d$ $(0,1)$-matrix $N$ and $d\times m$ $(0,1)$-matrix $L$ such that $M=NL$. 
It is also equal to the minimal number of bipartite cliques needed to partition all the edges of a bipartite graph whose adjacent matrix is $M$. The $s$-binary rank of $M$ is the minimal number of bipartite cliques needed to cover all edges of a bipartite graph whose adjacent matrix is $M$, where each edge is covered by at most $s$ bipartite cliques.
In \cite{ChalermsookHHK14}, it was shown that it is NP-hard to approximating the binary rank to within a factor of $n^{1-\delta}$ for any given $\delta$. 

A {\it property-testing} algorithm (tester) of the $s$-binary rank~\cite{ParnasRS21} is given as input $0<\epsilon<1$, integers $d,n,m$, and query access to the entries of a $n\times m$ $(0,1)$-matrix $M$. If $M$ has $s$-binary rank at most $d$ (resp. equal $d$), then the tester accepts with probability at least $2/3$. If $M$ is $\epsilon$-far from having $s$-binary rank at most $d$ (resp. equal $d$), i.e., more than $\epsilon$-fraction of the entries of $M$ should be modified to get a matrix with $s$-binary rank at most $d$ (resp. equal to $d$), then the tester rejects with probability at least $2/3$. If the tester
accepts matrices having $s$-binary rank at most $d$ (resp. equal to $d$) with probability $1$, then we call it a {\it one-sided error tester}. In {\it adaptive testing}, the queries can depend on the answers to the previous queries, whereas in {\it non-adaptive testing}, all the queries are fixed in advance by the tester. The goal is to construct a tester that makes a minimal number of queries. 

The testability of $s$-binary rank at most $d$ of $(0,1)$-matrices was studied in~\cite{NakarR18,ParnasRS21}. In~\cite{NakarR18}, Nakar and Ron gave a non-adaptive one-sided error tester for $s=1$, that makes $\tilde O(2^{4d}/\epsilon^4)$. In~\cite{ParnasRS21}, Parnas et al. gave a non-adaptive and adaptive one-sided error tester for $s=1$ that makes $O(2^{2d}/\epsilon^2)$ and $O(2^{2d}/\epsilon)$ queries, respectively. The results in~\cite{ParnasRS21} also hold for $s$-binary rank at most $d$. In this paper, for $s$-binary at most $d$ and equal to $d$, we prove
\begin{theorem}\label{TH1}
There exists an adaptive one-sided error tester for $s$-binary rank of $n\times m$ $(0,1)$-matrices that makes $\tilde O\left({d\choose \le s}2^d/\epsilon\right)$ queries.
\end{theorem}

\begin{theorem}\label{TH2}
There exists a non-adaptive one-sided error tester for $s$-binary rank of $n\times m$ $(0,1)$-matrices that makes $\tilde O\left({d\choose \le s}2^d/\epsilon^2\right)$ queries.
\end{theorem}
For fixed $s$, this improves the query complexity of Parnas et al. in~\cite{ParnasRS21} by a factor of $\tilde O(2^d)$.

\subsection{Our Approach}
The tester of Parnas et al.~\cite{ParnasRS21} uses the fact that if $M'$ is a $k\times k$ sub-matrix of $M$ and $M'$ is of $s$-binary rank at most $d$, then 
\begin{enumerate}
    \item $M'$ has at most $2^d$ distinct rows and at most $2^d$ distinct columns.
    \item If $M$ is $\epsilon$-far from having $s$-binary rank at most $d$, then extending $M'$ by one more uniformly at random row and column of $M$, gives a $(k+1)\times(k+1)$ sub-matrix $M''$ of $M$ that, with probability at least $\Omega(\epsilon)$, satisfies: the number of distinct rows in $M''$ is greater by one than the number of distinct rows in $M'$, or, the number of distinct columns in $M''$ is greater by one than the number of distinct columns in $M'$.
\end{enumerate} 
So, their adaptive tester runs  $O(2^d/\epsilon)$ iterations. At each iteration, it extends $M'$ by uniformly at random one row and one column. Let $M''$ be the resulting sub-matrix. If the $s$-binary rank of $M''$ is greater than $d$, the tester rejects. If the number of distinct rows or columns in $M''$ is greater than the number in $M'$, then it continues to the next iteration with $M'\gets M''$. Otherwise, it continues to the next iteration with $M'$. If, after $O(2^d/\epsilon)$ iterations, $M'$ has $s$-binary rank $d$, the tester accepts.

If the $s$-binary rank of $M$ is $d$, then every sub-matrix has a $s$-binary rank $d$, and the tester accepts. If $M$ is $\epsilon$-far  
from having $s$-binary rank at most $d$, then: since, at each iteration, with probability at least $\Omega(\epsilon)$, the number of distinct rows or columns of $M'$ is increased by one, and since matrices of $s$-binary rank $d$ has at most $2^d$ distinct rows and at most $2^d$ distinct columns, with high probability, we get $M'$ with $s$-binary rank greater than $d$ and the tester rejects. The query complexity of the tester is $O(2^{2d}/\epsilon)$, which is the number of entries of the matrix $M'$, $O(2^{2d})$, times the number of trials $O(1/\epsilon)$ for extending $M'$ by one row and one column.

We now give our approach. Call a sub-matrix $M'$ of $M$ {\it perfect} if it has distinct rows and distinct columns.
Our adaptive tester uses the fact that if $M'$ is a perfect $k\times k'$ sub-matrix of $M$ of $s$-binary rank $d$, then 
\begin{enumerate}
    \item\label{I1} $kk'\le {d\choose \le s}2^d$.
    \item\label{I2} If $M$ is $\epsilon$-far from having $s$-binary rank at most $d$, then at least one of the following occurs 
    \begin{enumerate}
        \item With probability at least $\Omega(\epsilon)$, extending $M'$ by one uniformly at random column of $M$, gives a perfect $k\times (k'+1)$ sub-matrix $M''$ of $M$.
        \item With probability at least $\Omega(\epsilon)$, extending $M'$ by one uniformly at random row of $M$, gives a perfect $(k+1)\times k'$ sub-matrix $M''$ of $M$.
        \item With probability at least $\Omega(\epsilon)$, extending $M'$ by one uniformly at random column and one uniformly at random row of $M$, gives a perfect\footnote{It may happen that events (a) and (b) do not occur and (c) does} $(k+1)\times (k'+1)$ sub-matrix $M''$ of $M$. 
    \end{enumerate}  
\end{enumerate} 
Item~\ref{I1} follows from Sgall result in~\cite{Sgall99} (See Section~\ref{Saglam}), and item~\ref{I2} is Claim 10 in~\cite{ParnasRS21}.
Now, the tester strategy is as follows. If $k\le k'$, the tester first tries to extend $M'$ with a new column. If it succeeds, it moves to the next iteration. Otherwise, it tries to extend $M'$ with a new row. If it succeeds, it moves to the next iteration. Otherwise, it tries to extend $M'$ with a new row and a new column. If it succeeds, it moves to the next iteration. If it fails, it accepts. 
If $k'<k$, it starts with the row, then the column, and then both. 

Using this strategy, we show that the query complexity will be, at most, the order of the size $kk'\le {d\choose \le s}2^d$ of $M'$ times the number of trials, $\tilde O(1/\epsilon)$, to find the new row, column, or both. This achieves the query complexity in Theorem~\ref{TH1}.

For the non-adaptive tester, the tester, uniformly at random, chooses $t=\tilde O\left({d\choose \le s}2^d/\epsilon^2\right)$ rows $r_1,\ldots,r_t\in [n]$ and $t$ columns $c_1,\ldots,c_t\in [m]$ and queries all $M[r_i,c_j]$ for all $i\cdot j\le t $ and puts them in a table. Then it runs the above non-adaptive tester. When the non-adaptive tester asks for uniformly at random row or column, it provides the next element $r_i$ or $c_j$, respectively.  The queries are then answered from the table. We show that the adaptive algorithm does not need to make queries that are not in the table before it halts. This achieves the query complexity in Theorem~\ref{TH2}.

\subsection{Other Rank Problems}
The {\it real rank} of a $n\times m$-matrix $M$ over any field $F$ is the minimal $d$, such that there is a $n\times d$ matrix $N$ over $F$ and a $d\times m$ matrix $L$ over $F$ such that $M=NL$.  The testability of the real rank was studied in~\cite{BalcanLW019,KrauthgamerS03,LiWW14}.  In~\cite{BalcanLW019}, Balcan et al. gave a non-adaptive tester for the real rank that makes $\tilde O(d^2/\epsilon)$ queries. They also show that this query complexity is optimal. 

The Boolean rank ($\infty$-binary rank) was studied in~\cite{NakarR18,ParnasRS21}. Parnas et al. in~\cite{ParnasRS21} gave a non-adaptive tester for the Boolean rank that makes $\tilde O(d^4/\epsilon^4)$ queries\footnote{The query complexity in~\cite{ParnasRS21} is $\tilde O(d^4/\epsilon^6)$. We've noticed that Lemma~3 in~\cite{ParnasRS21} is also true when we replace $(\epsilon^2/64)n^2$ with $(\epsilon/4)n^2$. To prove that, in the proof of Lemma 3, replace Modification rules 1 and 2 with the following modification:  Modify to $0$ all beneficial entries. This gives the result stated here,\cite{DanaPC}.}.

\section{Definitions and Preliminary Results}
Let $M$ be a $n\times m$ $(0,1)$-matrix. We denote by $R(M)$ and $C(M)$ the set of rows and columns of~$M$, respectively. The number of distinct rows and columns of $M$ are denoted by $r(M)=|R(M)|$ and, $c(M)=|C(M)|$, respectively. The {\it binary rank} of a $n\times m$-matrix $M$, $\brank(M)$, is equal to the minimal $d$, where there is a $n\times d$ $(0,1)$-matrix $N$ and a $d\times m$ $(0,1)$-matrix $L$ such that $M=NL$. 

We define the $s$-binary rank, $\br_s(M)$, of $M$ to be the minimal integer~$d$ such that there are $d$ sets (rectangles) $I_k\times J_k$ where $I_k\subseteq [n]:=\{1,\ldots,n\},J_k\subset [m],k\in [d]$ such that $M[i,j]=1$ for all $(i,j)\in I_k\times J_k$, $k\in [d]$ (monochromatic rectangles) and for every $(i,j)\in [n]\times [m]$ where $M[i,j]=1$ there are at least one and at most $s$ integers $t\in [d]$ such that $(i,j)\in I_t\times J_t$ (each~$1$-entry in $M$ is covered by at least one and at most $s$ monochromatic rectangles).

We now prove.
\begin{lemma}\label{DDD}
Let $M$ be a $n\times m$ $(0,1)$-matrix. The $s$-binary rank of $M$, $\br_s(M)$, is equal to the minimal integer $d$, where there is a $n\times d$ $(0,1)$-matrix $N$ and a $d\times m$ $(0,1)$-matrix $L$ such that: For $P=NL$,
\begin{enumerate}
    \item For every $(i,j)\in [n]\times [m]$, $M[i,j]=0$ if and only if $P[i,j]=0$.
    \item For every $(i,j)\in [n]\times [m]$, $P[i,j]\le s.$
\end{enumerate}
\end{lemma}
\begin{proof}
If $M$ is of $s$-binary rank $d$, then there are rectangles $\{I_k\times J_k\}_{k\in [d]}$, $I_k\subseteq [n],J_k\subset [m],k\in [d]$ such that $M[i,j]=1$ for all $(i,j)\in I_k\times J_k$, $k\in [d]$ and for every $(i,j)\in [n]\times [m]$ 
where $M[i,j]=1$ there are at least one and at most $s$ integers $t\in [d]$ such that $(i,j)\in I_t\times J_t$. Define row vectors $a^{(k)}\in \{0,1\}^n$ and $b^{(k)}\in \{0,1\}^m$ where $a^{(k)}_i=1$ iff (if and only if) $i\in I_k$, and $b^{(k)}_j=1$ iff $j\in J_k$. Then define\footnote{Here $x'$ is the transpose of $x$.}
$P={a^{(1)}}'b^{(1)}+\cdots+{a^{(d)}}'b^{(d)}$. It is easy to see that $({a^{(k)}}'b^{(k)})[i,j]=1$ iff
$(i,j)\in I_k\times J_k$. Therefore, $P[i,j]=0$ iff $M[i,j]=0$ and $P[i,j]\le s$ for all $(i,j)\in [n]\times [m]$. Define the $n\times d$ matrix $N=\left[{a^{(1)}}'|\cdots|{a^{(d)}}'\right]$ and the $d\times m$ matrix $L=\left[{b^{(1)}}'|\cdots|{b^{(d)}}'\right]'$. It is again easy to see that $P=NL$.

The other direction can be easily seen by tracing backward in the above proof.
\end{proof}

We now prove the following,
\begin{lemma}\label{drc}
Let $M$ be a $n\times m$ matrix. Let $N$ and $L$ be $n\times d$ $(0,1)$-matrix and $d\times m$ $(0,1)$-matrix, respectively, such that $P=NL$. Then $r(P)\le r(N)$ and $c(P)\le c(L)$.
\end{lemma}
\begin{proof}
We prove the result for $r$. The proof for $c$ is similar. Let $r_1,\ldots,r_n$ be the rows of $N$ and $p_1,\ldots,p_n$ be the rows of $P$. Then $p_i=r_iL$. Therefore, if $r_i=r_j$, then $p_i=p_j$. Thus, $r(P)\le r(N)$.
\end{proof}

Let $M$ be a $n\times m$ matrix. For $x\in X\subseteq [n]$, $y\in Y\subseteq [m]$, we denote by $M[X,Y]$ the $|X|\times |Y|$ sub-matrix of $M$,    $(M[x',y'])_{x'\in X,y'\in Y}$. Denote by $M[X,y]$ the column vector $(M[x',y])_{x'\in X}$ and by $M[x,Y]$ the row vector $(M(x,y'))_{y'\in Y}$. 

For $x\in [n]$ (resp. $y\in [m]$) we say that $M[X,y]$ is a {\it new column} (resp. $M[x,Y]$ is a {\it new row}) to $M[X,Y]$ if it is not equal to any of the columns (resp. rows) of $M[X,Y]$. 

\begin{lemma}\label{ngivesn}
Let M be a $n\times m$ matrix, $x\in [n], X\subseteq [n]$, $y\in [m]$, and $Y\subseteq [m]$. Suppose $M[x,Y]$ is not a new row to $M[X,Y]$, and $M[X,y]$ is not a new column to $M[X,Y]$. Then $M[x,Y\cup\{y\}]$ is not a new row to $M[X,Y\cup\{y\}]$ if and only if $M[X\cup\{x\},y]$ is not a new column to $M[X\cup\{x\},Y]$. 
\end{lemma}
\begin{proof}
If $M[x,Y\cup \{y\}]$ is not a new row to $M[X,Y\cup\{y\}]$, then there is $x'\in X$ such that $M[x,Y\cup\{y\}]=M[x',Y\cup\{y\}]$. Since $M[X,y]$ is not a new column to $M[X,Y]$, there is $y'\in Y$ such that $M[X,y]=M[X,y']$. Since $M[x,Y\cup \{y\}]=M[x',Y\cup \{y\}]$, we have $M[x',y']=M[x,y']$ and $M[x,y]=M[x',y]$. Since $M[X,y]=M[X,y']$, we have $M[x',y]=M[x',y']$. Therefore, $M[x,y]=M[x,y']$ and  $M[X\cup\{x\},y]=M[X\cup\{x\},y']$. Thus, $M[X\cup\{x\},y]$ is not a new column to $M[X\cup \{x\},Y]$.

Similarly, the other direction follows. 
\end{proof}

\section{Matrices of $s$-Binary Rank $d$}\label{Saglam}
In this section, we prove the following two Lemmas.
\begin{lemma}\label{rank}
For any $n\times m$ $(0,1)$-matrix $M$ of $s$-binary rank at most $d$, we have
$$r(M)\cdot c(M)\le {d\choose \le s}2^d.$$
\end{lemma}
\begin{lemma}\label{Eq}
There is a $(0,1)$-matrix $M'$ of $s$-binary rank $d$ that satisfies $r(M')\cdot c(M')={d\choose \le s}2^d.$
\end{lemma}

To prove Lemma~\ref{rank}, we use the following Sgall's lemma. 
\begin{lemma}\cite{Sgall99}.\label{kkeeyy}
Let $\A,\B\subseteq 2^{[d]}$ be such that for every $A\in \A$ and $B\in \B$, $|A\cap B|\le s$. Then $|\A|\cdot|\B|\le {d\choose \le s}2^d.$
\end{lemma}

We now prove Lemma~\ref{rank}.
\begin{proof}
Since the $s$-binary rank of $M$ is at most $d$, by Lemma~\ref{DDD}, there is a $n\times d$ $(0,1)$-matrix $N$ and a $d\times m$ $(0,1)$-matrix $L$ such that, for $P=NL$
\begin{enumerate}
    \item For every $(i,j)\in [n]\times [m]$, $M[i,j]=0$ if and only if $P[i,j]=0$.
    \item For every $(i,j)\in [n]\times [m]$, $P[i,j]\le s.$
\end{enumerate}
Obviously, $r(M)\le r(P)$ and $c(M)\le c(P)$. 
Consider $\A=\{A_1,\ldots,A_n\}\subseteq 2^{[d]}$ and $\B=\{B_1,\ldots,B_m\}\subseteq 2^{[d]}$, where $A_i=\{j|N_{i,j}=1\}$ and $B_k=\{j|L_{j,k}=1\}$. Since the entries of $P=NL$ are at most $s$, for every $i\in [n]$ and $k\in [m]$, $|A_i\cap B_k|\le s$.

By Lemma~\ref{drc} and~\ref{kkeeyy},
$$r(M)\cdot c(M)\le r(P)\cdot c(P)\le r(N)\cdot c(L)=  |\A|\cdot|\B|\le  {d\choose \le s}2^d.$$
\end{proof}

We now prove Lemma~\ref{Eq}

\begin{proof}
Let $N$ be a $2^d\times d$ $(0,1)$-matrix where its rows contain all the vectors in $\{0,1\}^d$. Let $L$ be a $d\times {d\choose \le s}$ matrix where its columns contain all the vectors in $\{0,1\}^d$ of weight at most $s$. Obviously, $P=NL$ is $2^d\times {d\choose \le s}$ with entries that are less than or equal to $s$. Define a $2^d\times {d\choose \le s}$ $(0,1)$-matrix $M'$ where $M'[i,j]=0$ if and only if $P[i,j]=0$. Then, by Lemma~\ref{DDD}, $M'$ is of $s$-binary rank at most~$d$. We now show that $r(M')\cdot c(M')= {d\choose \le s}2^d$. 

Since the identity $d\times d$ matrix $I_d$ is a sub-matrix of $L$, we have that $NI_d=N$ is $(0,1)$-matrix and a sub-matrix of $P$ and therefore of $M'$. Therefore, $r(M')\ge r(N)=2^d$. Since $I_d$ is a sub-matrix of $N$, by the same argument, $c(M')\ge c(L)={d\choose \le s}$. Therefore $r(M')\cdot c(M')\ge {d\choose \le s}2^d$. Thus, $r(M')\cdot c(M')= {d\choose \le s}2^d$.

We now show that $M'$ has $s$-binary rank $d$. Suppose the contrary, i.e.,  $M'$ has binary rank $d'<d$. Then there are $2^d\times d'$ $(0,1)$-matrix $N$ and $d'\times {d\choose \le s}$ $(0,1)$-matrix $L$ such that $P=NL$ and $M'[i,j]=0$ iff $P[i,j]=0$. Now by Lemma~\ref{drc}, $r(M')\le r(P)\le r(N)\le 2^{d'}<2^d$, which gives a contradiction. 
\end{proof}

\section{Testing The $s$-Binary Rank}
In this section, we present the  adaptive and non-adaptive testing algorithms for $s$-binary rank at most $d$. We first give the adaptive algorithm and prove Theorem~\ref{TH1}.

\subsection{The Adaptive Tester}

\newcounter{ALg}
\setcounter{ALg}{0}
\newcommand{\stepal}{\stepcounter{ALg}$\arabic{ALg}.\ $\>}
\newcommand{\steplabelal}[1]{\addtocounter{ALg}{-1}\refstepcounter{ALg}\label{#1}}
\begin{figure}[h!]
  \begin{center}
  \fbox{\fbox{\begin{minipage}{28em}
  \begin{tabbing}
  xx\=xxx\=xxxx\=xxxx\=xxxx\=xxxx\=xxxx \kill
  {{\bf Adaptive-Test-Rank}$(d,s,M,n,m,\epsilon)$}\\
{\bf Input}: Oracle that accesses the entries of $n\times m$ $(0,1)$-matrix $M$.  \\
{\bf Output}: Either ``Accept'' or ``Reject''\\
\\
\stepal\steplabelal{ALg01}
$X\gets \{1\}$; $Y\gets \{1\}$; $t= 9d/\epsilon$.\\
\stepal\steplabelal{ALg02}
{\bf While} $|X|\cdot |Y|\le {d\choose \le s}2^d$ {\bf do}\\
\stepal\steplabelal{ALg03}
\> {\bf If} the $s$-binary rank of $M[X,Y]$ is greater than $d$, {\bf then} Reject.\\
\stepal\steplabelal{ALg04}
\>$Finish\gets False$; $X'\gets \O$; $Y'\gets \O$. $/*$ $X'$ and $Y'$ are multi-sets.\\
\stepal\steplabelal{ALg05}
\> {\bf If} $|X|\ge |Y|$ {\bf then} \\
\stepal\steplabelal{ALg06}
\>\> {\bf While} (NOT $Finish$) AND $|X'|<t$ \\ \stepal\steplabelal{ALg07}
\>\>\> Draw uniformly at random  $x\in [n]\backslash X$; $X'\gets X'\cup\{x\}$;\\
\stepal\steplabelal{ALg08}
\>\>\> {\bf If} $M[x,Y]$ is a new row to $M[X,Y]$ {\bf then} $X\gets X\cup\{x\}$; $Finish\gets True.$\\
\stepal\steplabelal{ALg09}
\>\> {\bf If} (NOT $Finish$) {\bf then}\\
\stepal\steplabelal{ALg10}
\>\>\> {\bf While} (NOT $Finish$) AND $|Y'|<t$\\ 
\stepal\steplabelal{ALg11}
\>\>\>\> Draw uniformly at random  $y\in [m]\backslash Y$; $Y'\gets Y'\cup \{y\}$.\\
\stepal\steplabelal{ALg12}
\>\>\>\> {\bf If} $M[X,y]$ is new column to $M[X,Y]$ {\bf then} $Y\gets Y\cup\{y\}$; $Finish\gets True.$\\
\stepal\steplabelal{ALg13}
\> {\bf Else} ($|X|< |Y|$)  \\
\stepal\steplabelal{ALg14}
\>\> {\bf While} (NOT $Finish$) AND $|Y'|<t$\\ \stepal\steplabelal{ALg15}
\>\>\> Draw uniformly at random  $y\in [m]\backslash Y$; $Y'\gets Y'\cup\{y\}$;\\
\stepal\steplabelal{ALg16}
\>\>\> {\bf If} $M[X,y]$ is a new column to $M[X,Y]$ {\bf then} $Y\gets Y\cup\{y\}$; $Finish\gets True.$\\
\stepal\steplabelal{ALg17}
\>\> {\bf If} (NOT $Finish$) {\bf then}\\
\stepal\steplabelal{ALg18}
\>\>\> {\bf While} (NOT $Finish$) AND $|X'|<t$\\ 
\stepal\steplabelal{ALg19}
\>\>\>\> Draw uniformly at random  $x\in [n]\backslash X$; $X'\gets X'\cup \{x\}$\\
\stepal\steplabelal{ALg20}
\>\>\>\> {\bf If} $M[x,Y]$ is a new row to $M[X,Y]$ {\bf then} $X\gets X\cup\{x\}$; $Finish\gets True.$\\
\stepal\steplabelal{ALg21}
\> {\bf While} (NOT $Finish$) AND $X'\not=\O$ {\bf do}\\
\stepal\steplabelal{ALg22}
\>\> Draw uniformly at random $x\in X'$ and $y\in Y'$\\ 
\stepal\steplabelal{ALg23}
\>\>\> {\bf If} $M[x,Y\cup \{y\}]$ is a new row to $M[X,Y\cup \{y\}]$ OR, equivalently, \\
\stepal\steplabelal{ALg230}
\>\>\>\> $M[X\cup\{x\},y]$ is a new column to $M[X\cup \{x\},Y]$ \\
\stepal\steplabelal{ALg231}
\>\>\>\>\>{\bf then} $X\gets X\cup \{x\}; Y\gets Y\cup \{y\}$; $Finish\gets True$.\\
\stepal\steplabelal{ALg232}
\>\>\>\>\>{\bf else} $X'\gets X'\backslash \{x\}; Y'\gets Y'\backslash \{y\}$.\\
\stepal\steplabelal{ALg24}
\> {\bf If} (NOT $Finish$) {\bf then} Accept\\
\stepal\steplabelal{ALg25}
Reject
  \end{tabbing}
  \end{minipage}}}
  \end{center}
	\caption{An adaptive tester for $s$-binary rank at most $d$.}
	\label{cC}
\end{figure}

In this section, we prove Theorem~\ref{TH1}.

Consider the tester {\bf Adaptive-Test-Rank} in Figure~\ref{cC}.
The tester, at every iteration of the main While-loop (step~\ref{ALg02}) has a set $X$ of rows of $M$ and a set $Y$ of columns of $M$. If $|X|\ge |Y|$ (step~\ref{ALg05}), the tester first tries to extend $M[X,Y]$ with a new column (steps~\ref{ALg06}-\ref{ALg08}). If it succeeds, it moves to the next iteration. Otherwise, it tries to extend $M[X,Y]$ with a new row (steps~\ref{ALg09}-\ref{ALg12}). If it succeeds,  it moves to the next iteration. Otherwise, it tries to extend $M[X,Y]$ with a new row and a new column (steps~\ref{ALg21}-\ref{ALg232}). If it succeeds, it moves to the next iteration. If it fails, it accepts (step~\ref{ALg24}). 
If $|X|< |Y|$ (step~\ref{ALg13}), it starts with the row of $M[X,Y]$ (steps~\ref{ALg14}-\ref{ALg16}), then the column (steps~\ref{ALg18}-\ref{ALg20}), and then both (steps~\ref{ALg21}-\ref{ALg232}). If it fails, it accepts (step~\ref{ALg24}). 

If $|X|\cdot|Y|>{d\choose \le s}2^d$ (step~\ref{ALg02} and then step~\ref{ALg25}) or the $s$-binary rank of $M[X,Y]$ is greater than $d$ (step~\ref{ALg03}), then it rejects.

We first prove
\begin{lemma} Let $t= 9d/\epsilon$.
Tester {\bf Adaptive-Test-Rank} makes at most $2{d\choose \le s}2^dt=\tilde O\left({d\choose \le s}2^d\right)/\epsilon$ queries.
\end{lemma}
\begin{proof}
We prove by induction that at every iteration of the main While-loop (step~\ref{ALg02}), the tester knows the entries of $M[X,Y]$, and the total number of queries, $q_{X,Y}$, is at most $2|X||Y|t$. Since the While-loop condition is $|X||Y|\le {d\choose \le s}2^d$, the result follows. 

At the beginning of the algorithm, no queries are made, and $|X|=|Y|=1$. Then $2|X||Y|t=2t>0=q_{X,Y}$. Suppose, at the $k$th iteration, the tester knows the entries of $M[X,Y]$ and $q_{X,Y}\le 2|X||Y|t$. We prove the result for the $(k+1)$th iteration. 

We have the following cases (at the $(k+1)$th iteration)

\noindent
{\bf Case I.} $|X|\ge |Y|$ (step \ref{ALg05}) and, for some $x$, $M[x,Y]$ is a new row to $M[X,Y]$ (step~\ref{ALg08}).

In that case, $Finish$ becomes $true$, and no other sub-while-loop is executed.
Therefore, the number of queries made at this iteration is at most $|Y|t$ (to find all $M[x,Y]$), and one element $x$ is added to $X$. Then, the tester knows all the entries of $M[X\cup\{x\},Y]$ and $$q_{X\cup\{x\},Y}=q_{X,Y}+|Y|t\le 2|X||Y|t+|Y|t\le 2|X\cup\{x\}|\cdot |Y|t,$$ and the result follows.

\noindent
{\bf Case II.} $|X|\ge |Y|$ (step \ref{ALg05}), for all $x'\in X'$, $M[x',Y]$ is not a new row to $M[X,Y]$ (step~\ref{ALg08}), and for some $y$, $M[X,y]$ is a new column to $M[X,Y]$ (step~\ref{ALg12}). 

In that case, $Finish$ becomes $true$, and no other sub-while-loop is executed after the second sub-while-loop (step~\ref{ALg10}).

Therefore, in this case, the number of queries made at this iteration is at most $|Y|t+|X|t$. $|X|t$ queries in the first sub-while-loop (to find $M[x,Y]$ for all $x\in X'$), and at most $|Y|t$ queries in the second sub-while-loop (to find $M[X,y']$ for all $y'\in Y'$). Then one element $y$ is added to $Y$. Therefore, the tester knows the entries of $M[X,Y\cup\{y\}]$ and, since $|Y|\le |X|$, $$q_{X,Y\cup\{y\}}=q_{X,Y}+|X|t+|Y|t\le 2|X||Y|t+2|X|t= 2|X|\cdot |Y\cup\{y\}|t,$$ and the result follows.

\noindent
{\bf Case III.} $|X|\ge |Y|$, for all $x'\in X'$, $M[x',Y]$ is not a new row to $M[X,Y]$, for all $y'\in Y'$, $M[X,y']$ is not a new column to $M[X,Y]$, and for some $x\in X', y\in Y'$, $M[x,Y\cup\{y\}]$ is a new row to $M[X,Y\cup \{y\}]$ (step~\ref{ALg23}). 

In this case, $|X'|=|Y'|=t$, the number of queries is $|X|t+|Y|t+t$. Exactly $|X|t$ queries in the first sub-while-loop, $|Y|t$ queries in the second sub-while-loop, and at most\footnote{This is because, for $x\in X',y\in Y'$, the tester already knows $M[x,Y]$ and $M[X,y]$ from the first and second sub-while-loop and only needs to query $M[x,y]$.}  $t$ queries in the sub-while-loop in step~\ref{ALg21}. 
Then 
one element $x$ is added to $X$, and one element $y$ is added to $Y$. Then the tester knows the entries of $M[X\cup \{x\},Y\cup\{y\}]$ and
$$q_{X\cup\{x\},Y\cup\{y\}}=q_{X,Y}+|X|t+|Y|t+t\le 2|X|\cdot|Y|t+|X|t+|Y|t+t\le  2|X\cup \{x\}|\cdot|Y\cup\{y\}|t.$$ 

\noindent
{\bf Case IV.} $|X|\ge |Y|$, for all $x'\in X'$, $M[x',Y]$ is not a new row to $M[X,Y]$, for all $y'\in Y'$, $M[X,y']$ is not a new column to $M[X,Y]$, and for all the drawn pairs $x\in X', y\in Y'$, $M[x,Y\cup\{y\}]$ is not a new row to $M[X,Y\cup \{y\}]$ (step~\ref{ALg23}).

In this case, $Finish$ will have value $False$, and the tester accepts in step~\ref{ALg24}.  

The analysis of the case when $|X|<|Y|$ is similar to the above analysis. 
\end{proof}

We now prove the completeness of the tester.
\begin{lemma}
If $M$ is a $n\times m$ $(0,1)$-matrix of $s$-binary rank at most $d$, then the tester {\bf Adaptive-Test-Rank} accepts with probability $1$.
\end{lemma}
\begin{proof}
The tester rejects if and only if one of the following occurs,
\begin{enumerate}
    \item $M[X,Y]$ has $s$-binary rank greater than $d$.
    \item $|X|\cdot |Y|>{d\choose \le s}2^d$.
\end{enumerate}
If $M[X,Y]$ has $s$-binary rank greater than $d$, then $M$ has $s$-binary rank greater than $d$. This is because, if $M=NL$, then $M[X,Y]=N[X,[d]]\cdot L[[d],Y]$. So item 1 cannot occur.

Before we show that item 2 cannot occur, we prove the following:
\begin{claim}
The rows (resp. columns) of $M[X,Y]$ are distinct. 
\end{claim}
\begin{proof}
The steps in the tester where we add rows or columns are steps~\ref{ALg08}, \ref{ALg12} \ref{ALg16}, \ref{ALg20}, and~\ref{ALg23}. 
In steps~\ref{ALg08}, \ref{ALg12} \ref{ALg16}, \ref{ALg20} it is clear that a row (resp. column) is added only if it is a new row (resp. column) to $M[X,Y]$. 
Consider step~\ref{ALg23} and suppose, w.l.o.g $|X|\ge |Y|$. This step is executed only when $Finish=False$. This happens when $|X'|=|Y'|=t$, for every $x\in X'$, $M[x,Y]$ is not a new row to $M[X,Y]$, and for every $y\in Y'$, $M[X,y]$ is not a new column to $M[X,Y]$. Then $x$ and $y$ are added to $X$ and $Y$, respectively, if $M[x,Y\cup \{y\}]$ is a new row to $M[X,Y\cup \{y\}]$. Then, by Lemma~\ref{ngivesn}, $M[X\cup\{x\},y]$ is a new column to $M[X\cup\{x\},Y]$. So, the rows (and columns) in  $M[X\cup\{x\},Y\cup\{y\}]$ are distinct. This implies the result.
\end{proof}
Suppose, to the contrary, $|X|\cdot |Y|>{d\choose \le s}2^d$. Since $M'=M[X,Y]$ satisfies $r(M')c(M')=|X|\cdot |Y|>{d\choose \le s}2^d$, by Lemma~\ref{rank}, the $s$-binary rank of $M'$, and therefore of $M$, is greater than $d$. A contradiction.
\end{proof}
We now prove the soundness of the tester.

We first prove the following.

\begin{claim}
Let $M$ be a $n\times m$ $(0,1)$-matrix, $X\subseteq [n]$, and $Y\subseteq [m]$. Suppose there are two functions, $':[n]\to X$ and $'':[m]\to Y$, such that
\begin{enumerate}
\item For every $x\in [n]$,  $M[x,Y]=M[x',Y]$.
\item For every $y\in [m]$, $M[X,y]=M[X,y'']$.
\item For every $x\in [n]$ and $y\in [m]$, $M[x,y]=M[x',y'']$.
\end{enumerate}
Then $M$ has at most $|X|$ distinct rows and $|Y|$ distinct columns, and its $s$-binary rank is the $s$-binary rank of $M[X,Y]$.
\end{claim}
\begin{proof}
Let $x\in [n]\backslash X$. For every $y$, $M[x,y]=M[x',y'']=M[x',y]$. Therefore, row $x$ in $M$ is equal to row $x'$. Similarly, column $y$ in $M$ is equal to column $y''$.

Since adding equal columns and rows to a matrix does not change the $s$-binary rank\footnote{If we add a column to a matrix that is equal to column $y$, then the rectangles that cover column $y$ can be extended to cover the added column.}, we have $\brank_s(M[X,Y])=\brank_s(M[X,[m]])=\brank_s(M)$.
\end{proof}

The following Claim is proved in~\cite{ParnasRS21} (Claim 10). Here, we give the proof for completeness.
\begin{claim}\label{ranks}Let $M$ be a $(0,1)$-matrix that is $\epsilon$-far from having $s$-binary rank at most $d$. Let $X\subseteq [n]$ and $Y\subseteq [m]$, such that $\brank_s(M[X,Y])\le d$, the columns of $M[X,Y]$ are distinct, and the rows of $M[X,Y]$ are distinct. Then one of the following must hold:
\begin{enumerate}
    \item The number of rows $x\in [n]$ where $M[x,Y]$ is a new row to $M[X,Y]$ is at least $ n\epsilon/3$.
    \item The number of columns $y\in [m]$ where $M[X,y]$ is a new column to $M[X,Y]$ is at least $ m\epsilon/3$.
    \item The number pairs $(x,y)$, $x\not\in X$, $y\not\in Y$, where, $M[x,Y]=M[x',Y]$ for some $x'\in X$,  $M[X,y]=M[X,y'']$ for some $y''\in Y$, and $M[x,y]\not=M[x',y'']$, is at least $ mn\epsilon/3$.
\end{enumerate}
\end{claim}
\begin{proof}
Assume, to the contrary, that none of the above statements holds.  Change every row $x$ in $M$ where $M[x,Y]$ is a new row to $M[X,Y]$ to a zero row. Let $X'$ be the set of such rows. Change every column $y$ in $M$ where $M[X,y]$ is a new row to $M[X,Y]$ to a zero column. Let $Y'$ be the set of such columns. For every other entry $(x,y)$, $x\not\in X$, $y\not\in Y$ that is not changed to zero and $M[x,y]\not=M[x',y'']$, change  $M[x,y]$ to $M[x',y'']$. Let $M'$ be the  matrix obtained from the above changes. 

The number of entries $(x,y)$ where $M[x,y]\not=M'[x,y]$ is less than $(n\epsilon/3)m+(m\epsilon/3)n+mn\epsilon/3=\epsilon mn$. Therefore, $M'$ is $\epsilon$-close to $M$. By claim~\ref{ranks}, $\brank_s(M')=\br_s(M[[n]\backslash X',[m]\backslash Y'])=\br_s(M[X,Y])\le d$. A contradiction. 
\end{proof}
We now prove the completeness of the tester.
\begin{lemma}
If $M$ is $\epsilon$-far from having $s$-binary rank $d$, then with probability at least $2/3$, {\bf Adaptive-Test-Rank} rejects. 
\end{lemma}
\begin{proof}
Consider the while-loop in step~\ref{ALg02} at some iteration $i$. If $\brank_s(M[X,Y])>d$, then the tester rejects in step~\ref{ALg03}. We will now show that if $\brank_s(M[X,Y])\le d$, then, with probability at most $3e^{-2d}$, the tester accepts at iteration $i$.

To this end, let $\brank_s(M[X,Y])\le d$. Then, by Claim~\ref{ranks}, one of the following holds.
\begin{enumerate}
    \item\label{it1} The number of rows $x\in [n]$ where $M[x,Y]$ is a new row to $M[X,Y]$ is at least $ n\epsilon/3$.
    \item\label{it2} The number of columns $y\in [m]$ where $M[X,y]$ is a new column to $M[X,Y]$ is at least $ m\epsilon/3$.
    \item\label{it3} The number pairs $(x,y)$, $x\not\in X$, $y\not\in Y$, where, $M[x,Y]=M[x',Y]$ for some $x'\in X$,  $M[X,y]=M[X,y'']$ for some $y''\in Y$, and $M[x,y]\not=M[x',y'']$, is at least $ mn\epsilon/3$.
\end{enumerate}
Now at the $i$th iteration, suppose w.l.o.g, $|X|\ge |Y|$ (the other case $|Y|<|X|$ is similar). If item~\ref{it1} occurs, then with probability at least $p=1-(1-\epsilon/3)^t\ge 1-e^{-2d}$, the tester finds a new row to $M[X,Y]$ and does not accept at iteration $i$. If item~\ref{it2} occurs, then if it does not find a new row to $M[X,Y]$, with probability at least $p$, the tester finds a new column to $M[X,Y]$ and does not accept. If item~\ref{it3} occurs, and it does not find a new row or column to $M[X,Y]$, then with probability at least $p$, it finds such a pair and does not accept. Therefore, with probability at most $3(1-p)\le 3e^{-2d}$, the tester accepts at iteration $i$. 

Since the while-loop runs at most $|X|+|Y|\le 2|X||Y|\le 2{d\choose \le s}2^d\le 2^{2d+1}$ iterations, with probability at most $3e^{-2d}2^{2d+1}\le 1/3$, the tester accepts in while-loop. 
Therefore, with probability at least $2/3$, the tester does not accept in the while-loop. Thus, it either rejects because $\brank_s(M[X,Y])>d$ or rejects in step~\ref{ALg25}.
\end{proof}

\subsection{The Non-Adaptive Tester}
In this section, we prove Theorem~\ref{TH2}.
\newcounter{ALC2}
\setcounter{ALC2}{0}
\newcommand{\stepb}{\stepcounter{ALC2}$\arabic{ALC2}.\ $\>}
\newcommand{\steplabelb}[1]{\addtocounter{ALC2}{-1}\refstepcounter{ALC2}\label{#1}}
\begin{figure}[h!]
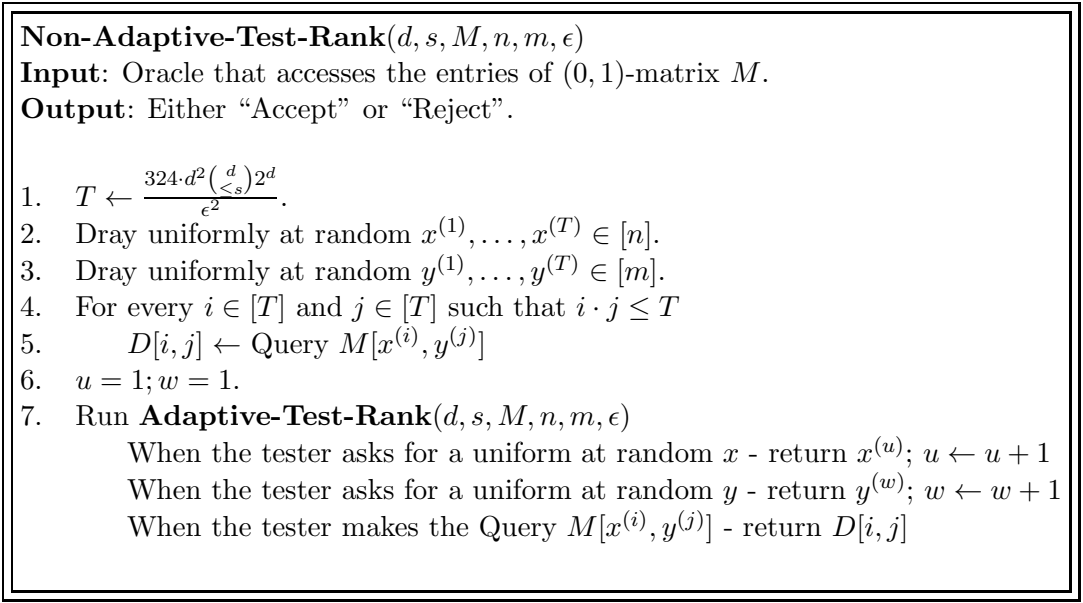

  \begin{center}
  \fbox{\fbox{\begin{minipage}{28em}
  \begin{tabbing}
  xxx\=xxxx\=xxxx\=xxxx\=xxxx\=xxxx\= \kill
  {{\bf Non-Adaptive-Test-Rank}$(d,s,M,n,m,\epsilon)$}\\
 {\bf Input}: Oracle that accesses the entries of $(0,1)$-matrix $M$. \\
  {\bf Output}: Either ``Accept'' or ``Reject''.\\ \\
\stepb\steplabelb{Lit01}
$T\gets \frac{324\cdot d^2{d\choose \le s}2^d}{\epsilon^2}$.\\
\stepb\steplabelb{LitT1}
Dray uniformly at random $x^{(1)},\ldots,x^{(T)}\in [n]$.\\
\stepb\steplabelb{Uni1}
Dray uniformly at random $y^{(1)},\ldots,y^{(T)}\in [m]$.\\
\stepb\steplabelb{Rej211}
For every $i\in [T]$ and $j\in [T]$ such that $i\cdot j\le T$\\
\stepb\steplabelb{ConB1}
\> $D[i,j]\gets$ Query $M[x^{(i)},y^{(j)}]$\\
\stepb\steplabelb{CowE11}
$u=1; w=1.$\\
\stepb\steplabelb{ConE11}
Run {\bf Adaptive-Test-Rank}$(d,s,M,n,m,\epsilon)$\\
\> \>When the tester asks for a uniform at random $x$ - return $x^{(u)}$; $u\gets u+1$\\
\> \>When the tester asks for a uniform at random $y$ - return $y^{(w)}$; $w\gets w+1$\\
\>\> When the tester makes the Query $M[x^{(i)},y^{(j)}]$ - return $D[i,j]$\\
  \end{tabbing}
  \end{minipage}}}
  \end{center}
	\caption{A non-adaptive tester for $s$-binary rank at most $d$.}
	\label{A32}
	\end{figure}

First, consider {\bf Adaptive-Test-Rank} in Figure~\ref{cC}. Consider steps~\ref{ALg07},\ref{ALg11},\ref{ALg15}, and~\ref{ALg19}, where it draws a new column or row. We prove.
\begin{lemma}\label{TDR} Let $t=9d/\epsilon$. 
At each iteration of {\bf Adaptive-Test-Rank}, the total number of  uniformly at random rows $x\in [n]$ drawn is at most $(|X|+\min(|X|,|Y|-1))t$, and the number of uniformly at random rows $y\in [m]$ drawn is at most $(|Y|+\min(|X|,|Y|))t$.
\end{lemma}
\begin{proof}
We prove by induction that at every iteration of the main While-loop (step~\ref{ALg02}), the total number of random rows drawn by the tester, $n_{X,Y}$, is at most $(|X|+\min(|X|,|Y|-1))t$, and the total number of random columns drawn, $m_{X,Y}$, is at most $(|Y|+\min(|X|,|Y|))t$. 

At the beginning, $|X|=|Y|=1$, and the number of columns and rows is $1$. In that case,\footnote{We assume that the first column/row drawn is column/row one}, $n_{X,Y}=1\le t$ and $m_{X,Y}=1\le 2t$. Suppose, at the $k$th iteration, the induction statement is true. We prove the result for the $(k+1)$th iteration. 

At the $(k+1)$th iteration, we have the following cases.

\noindent
{\bf Case I.} $|X|\ge |Y|$ (step \ref{ALg05}) and, for some $x$, $M[x,Y]$ is a new row to $M[X,Y]$ (step~\ref{ALg08}).

In that case, $Finish$ becomes $true$, and no other sub-while-loop is executed.
Therefore, the number of rows drawn at this iteration is at most $t$, and one element $x$ is added to $X$. No columns are drawn. Then,  $$n_{X\cup\{x\},Y}\le n_{X,Y}+t\le (|X|+\min(|X|,|Y|-1)+1)t\le (|X\cup\{x\}|+\min(|X\cup\{x\}|,|Y|-1))t,$$ and $$m_{X\cup\{x\},Y}=m_{X,Y}\le (|Y|+\min(|X|,|Y|))t\le (|Y|+\min(|X\cup\{x\}|,|Y|))t.$$
Thus, the result follows for this case.

\noindent
{\bf Case II.} $|X|\ge |Y|$ (step \ref{ALg05}), for all $x'\in X'$, $M[x',Y]$ is not a new row to $M[X,Y]$ (step~\ref{ALg08}), and for some $y$, $M[X,y]$ is a new column to $M[X,Y]$ (step~\ref{ALg12}). 

In that case, $Finish$ becomes $true$, and no other sub-while-loop is executed after the second sub-while-loop (step~\ref{ALg10}).

Therefore, in this case, the number of rows drawn at this iteration is $t$, one element $y$ is added to $Y$, and the number of columns drawn is at most $t$. Then
\begin{eqnarray*}
n_{X,Y\cup\{y\}}=n_{X,Y}+t &\le& (|X|+\min(|X|,|Y|-1)+1)t\\
&=& (|X|+|Y|)t= (|X|+\min(|X|,|Y\cup\{y\}|-1))t,
\end{eqnarray*} and
$$m_{X,Y\cup\{y\}}\le m_{X,Y}+t\le (|Y|+\min(|X|,|Y|)+1)t\le (|Y\cup\{y\}|+\min(|X|,|Y\cup\{y\}|))t.$$ 
Thus, the result follows for this case.

\noindent
{\bf Case III.} $|X|< |Y|$ (step~\ref{ALg13}), and for some $y$, $M[X,y]$ is a new column to $M[X,Y]$ (step~\ref{ALg16}).

In that case, $Finish$ becomes $true$, and no other sub-while-loop is executed.
Therefore, the number of columns drawn at this iteration is at most $t$, and one element $y$ is added to $Y$. No rows are drawn. Then,  $$n_{X,Y\cup\{y\}}=n_{X,Y}\le (|X|+\min(|X|,|Y|-1))t\le (|X|+\min(|X|,|Y\cup\{y\}|-1))t,$$ and $$m_{X,Y\cup\{y\}}\le m_{X,Y}+t\le  (|Y|+\min(|X|,|Y|)+1)t= (|Y\cup\{y\}|+\min(|X|,|Y\cup\{y\}|))t.$$
Thus, the result follows for this case.

\noindent
{\bf Case IV.} $|X|< |Y|$ (step~\ref{ALg13}), for all $y'\in Y'$, $M[X,y']$ is not a new row to $M[X,Y]$, and for some $x$, $M[x,Y]$ is a new column to $M[X,Y]$ (step~\ref{ALg20}). 
In that case, $Finish$ becomes $true$, and no other sub-while-loop is executed after the fourth sub-while-loop (step~\ref{ALg18}).

In this case, the number of rows drawn at this iteration is $t$, one element $x$ is added to $X$, and the number of columns drawn is at most $t$. Then
\begin{eqnarray*}
n_{X\cup\{x\},Y}=n_{X,Y}+t &\le& (|X|+\min(|X|,|Y|-1)+1)t\\
&\le & (|X\cup\{x\}|+\min(|X\cup\{x\}|,|Y|-1))t
\end{eqnarray*}
$$m_{X\cup\{x\},Y}\le m_{X,Y}+t\le (|Y|+\min(|X|,|Y|)+1)t= (|Y|+\min(|X\cup\{x\}|,|Y|))t.$$ 
Thus, the result follows for this case.

\noindent
{\bf Case V.} For all $x'\in X'$, $M[x',Y]$ is not a new row to $M[X,Y]$, for all $y'\in Y'$, $M[X,y']$ is not a new column to $M[X,Y]$, and for some $x\in X', y\in Y'$, $M[x,Y\cup\{y\}]$ is a new row to $M[X,Y\cup \{y\}]$ (step~\ref{ALg23}). 

In this case, the number of rows drawn at this iteration is $t$, the number of columns drawn is $t$, one element $x$ is added to $X$, and one element $y$ is added to $Y$.
Then
\begin{eqnarray*}
n_{X\cup\{x\},Y\cup\{y\}}=n_{X,Y}+t &\le& (|X|+\min(|X|,|Y|-1)+1)t\\
&\le& (|X\cup\{x\}|+\min(|X\cup\{x\}|,|Y\cup\{y\}|-1))t.
\end{eqnarray*}
\begin{eqnarray*}
m_{X\cup\{x\},Y\cup\{y\}}=m_{X,Y}+t &\le& (|Y|+\min(|X|,|Y|)+1)t\\
&\le& (|Y\cup\{y\}|+\min(|X\cup\{x\}|,|Y\cup\{y\}|))t.
\end{eqnarray*}
\end{proof}

We are now ready to prove Theorem~\ref{TH2}.
\begin{proof}
By Lemma~\ref{TDR}, the total number of rows and columns drawn in {\bf Adaptive-Test-Rank} up to iteration $t$ is at most $n':=9(|X|+\min(|X|,|Y|-1))d/\epsilon\le 18|X|d/\epsilon$ and $m':=9(|Y|+\min(|X|,|Y|)d/\epsilon\le 18|Y|d/\epsilon$, respectively. We also have $|X|\cdot |Y|\le {d\choose \le s}2^d$. So $$n'\cdot m'\le 324|X||Y|d^2/\epsilon^2\le T:= \frac{324\cdot d^2{d\choose \le s}2^d}{\epsilon^2}.$$ 

Consider the tester {\bf Non-Adaptive-Test-Rank} in Figure~\ref{A32}. The tester draws $T$ rows $x^{(1)},\ldots,$ $x^{(T)}\in [n]$, and columns $y^{(1)},\ldots,y^{(T)}\in [m]$ and queries all $M[x^{(i)},y^{(j)}]$ where $ij\le T$ and puts the result in the table $D$. Then it runs {\bf Adaptive-Test-Random} using the above-drawn rows and columns. We now show that all the queries that {\bf Adaptive-Test-Random} makes can be fetched from the table~$D$.

At any iteration, the number of rows drawn is at most $n'$, and the number of rows drawn is at most $m'$. Therefore, the tester needs to know (in the worst case) all the entries $M[x^{(i)},y^{(j)}]$ where $i\le n'$ and $j\le m'$. Since $ij\le n'm'\le T$, the result follows.

The number of queries that the tester makes is 
$$\sum_{i=1}^T\frac{T}{i}=O(T\ln T)=\tilde O\left(\frac{{d\choose \le s}2^d}{\epsilon^2}\right).$$
\end{proof}

\section{Testing the Exact $s$-Binary Rank}
We first prove the following.
\begin{lemma}\label{ChRC}
Let $M$ and $M'$ be $n\times m$ $(0,1)$-matrices that differ in one row (or column). Then $|\brank_s(M)-\brank_s(M')|\le 1$.
\end{lemma}
\begin{proof} 
Suppose $\brank_s(M)=d$ and $M'$ differ from $M$ in row $k$. Let $N$ and $L$ be $n\times d$ $(0,1)$-matrix and $d\times m$ $(0,1)$-matrix, respectively, such that $P=NL$, for every $(i,j)\in[n]\times [m]$, $P[i,j]\le s$, and $P[i,j]=0$ if and only if $M[i,j]=0$. Add to $N$ a column (as a $(d+1)$th column) that all its entries are zero except the $k$-th entry, which equals $1$. Then change $N[k,j]$ to zero for all $j\in [d]$. Let $N'$ be the resulting matrix. Add to $L$ another row (as a $(d+1)$th row) equal to the $k$-th row of $M'$. Let $L'$ be the resulting matrix. Let $P'=N'L'$. It is easy to see that $P'[i,j]=P[i,j]$ for all $i\not=k$ and $j$, and the $k$th row of $P'$ is equal to the $k$th row of $M'$. Then, for every $(i,j)\in[n]\times [m]$, $P'[i,j]\le s$, and $P'[i,j]=0$ if and only if $M'[i,j]=0$. Therefore, $\brank_s(M')\le d+1=\brank_s(M)+1$. In the same way, $\brank_s(M)\le \brank_s(M')+1$.
\end{proof}

\begin{lemma}\label{EEE} Let $\eta=d^2/(nm)$.
Let $M$ be $n\times m$ $(0,1)$-matrix. If $M$ is $\epsilon$-close to having $s$-binary rank at most $d$, then $M$ is $(\epsilon+\eta)$-close to having $s$-binary rank $d$. 
\end{lemma}
\begin{proof} 
We will show that for every $n\times m$ $(0,1)$-matrix $H$ of $s$-binary rank at most $d-1$, there is a $n\times m$ $(0,1)$-matrix $G$ of $s$-binary rank $d$ that is $\eta$-close to $H$. Therefore, if $M$ is $\epsilon$-close to having $s$-binary rank at most $d$, then it is $(\epsilon+\eta)$-close to having $s$-binary rank $d$. 

Define the $n\times m$ $(0,1)$-matrices $G_k$, $k\in [d]\cup\{0\}$, where $G_0=H$ and for $k\ge 1$, $G_k[i,j]=H[i,j]$ if $j>k$ or $i>d$, and $G_k[[d],[k]]=I_d[[d],[k]]$ where $I_d$ is the $d\times d$ identity matrix. Since $G_d[[d],[d]]=I_d$, we have $\br_s(G_d)\ge d$. It is clear that for every $k\in [d]\cup \{0\}$, $G_k$ is $(d^2/nm)$-close to $H$. If $\br_s(G_d)= d$, then take $G=G_d$, and we are done. Otherwise, suppose $\br_s(G_d)>d$. 

Now consider a sequence $H=G_0,G_1,G_2,\ldots,G_d$. By Lemma~\ref{ChRC}, we have $\brank_s(G_{i-1})-1\le \brank_s(G_i)\le \brank_s(G_{i-1})+1$. Now since $\brank_s(G_0)=\brank_s(H)\le d-1$ and $\brank_s(G_d)>d$, by the discrete intermediate value theorem, there must be $k\in [d]$ such that $\brank_s(G_k)=d$. Then take $G=G_k$, and we are done.
\end{proof}

Now, the tester for testing the $s$-binary rank $d$ runs as follows. If $mn<2d^2/\epsilon$, then find all the entries of $M$ with $mn<2d^2/\epsilon$ queries. If $\br_s(M)=d$, then accept. Otherwise, reject.
If $mn\ge 2d^2/\epsilon$, then run {\bf Adaptive-Test-Rank}$(d,s,M,n,m,\epsilon/2)$ (for the non-adaptive, we run {\bf Non-Adaptive-Test-Rank}$(d,s,M,n,m,\epsilon/2)$) and output its answer.

We now show the correctness of this algorithm. If $M$ is of $s$-binary rank $d$, then it is of $s$-binary rank at most $d$, and the tester accepts.

Now, suppose $f$ is $\epsilon$-far from having $s$-binary rank $d$. If $mn<2d^2/\epsilon$, the tester rejects. If $mn\ge 2d^2/\epsilon$, then, by Lemma~\ref{EEE}, $f$ is $(\epsilon-\eta)$-far from  having $s$-binary rank at most $d$, where $\eta=d^2/(nm)$. Since $\eta=d^2/(nm)\le \epsilon/2$, the function $f$ is $(\epsilon/2)$-far from  having $s$-binary rank at most $d$, and therefore the tester, with probability at least $2/3$, rejects.

\bibliography{TestingRef}

\newpage

\end{document}